\documentclass[aps,pra,reprint,floatfix,superscriptaddress]{revtex4-1}
\pdfoutput=1
\usepackage[ascii]{inputenc}
\usepackage[bookmarksnumbered,hypertexnames=false]{hyperref}
\usepackage{bbm,braket,microtype,mathrsfs,amsmath,amssymb,color,amsthm,graphicx,cleveref}
\usepackage[T1]{fontenc}
\usepackage[USenglish]{babel}
\usepackage{times}
\usepackage{verbatim}
\usepackage{soul}

\newtheorem{thm}{Theorem}\crefname{thm}{theorem}{theorems}
\newtheorem{lem}[thm]{Lemma}\crefname{lem}{lemma}{lemmas}
\crefname{cor}{corollary}{corollaries}
\crefname{dfn}{definition}{definitions}
\DeclareMathOperator{\tr}{tr}

\newcommand{\ot}{\otimes}

\newcommand{\Ep}{\mathcal E}
\newcommand{\Dp}{\mathcal D}
\newcommand{\ketum}{\ket {{\phi^+}}}
\newcommand{\braum}{\bra {{\phi^+}}}
\newcommand{\ssection}[1]{\smallskip\phantomsection\addcontentsline{toc}{section}{#1}\textit{#1.---}}

\allowdisplaybreaks[4]

\begin{document}

%%%%%%%%%%%%%%%%%%%%%%%%%%%%%%%%%%%%%%%%%%%%%%%%%%%%%%%%%%%%%%%%%%%%%%%%%%%%%%
\title{Error Correction of Quantum Reference Frame Information}
\author{Patrick Hayden}
\affiliation{Stanford Institute for Theoretical Physics, Stanford University, Stanford, California 94305, USA}
\author{Sepehr Nezami}
\affiliation{Stanford Institute for Theoretical Physics, Stanford University, Stanford, California 94305, USA}
\author{Sandu Popescu}
\affiliation{H. H. Wills Physics Laboratory, University of Bristol, Tyndall Avenue, Bristol, BS8 1TL, United Kingdom}
\author{Grant Salton}
\affiliation{Stanford Institute for Theoretical Physics, Stanford University, Stanford, California 94305, USA}
\begin{abstract}
The existence of quantum error correcting codes is one of the most counterintuitive and potentially technologically important discoveries of quantum information theory. However, standard error correction refers to abstract quantum information, \emph{i.e.}, information that is independent of the physical incarnation of the systems used for storing the information. There are, however, other forms of information that are \emph{physical} -- one of the most ubiquitous being reference frame information. Here we analyze the problem of error correcting physical information. The basic question we seek to answer is whether or not such error correction is possible and, if so, what limitations govern the process. The main challenge is that the systems used for transmitting physical information, in addition to any actions applied to them, must necessarily obey these limitations.  Encoding and decoding operations that obey a restrictive set of limitations need not exist \emph{a priori}. We focus on the case of erasure errors, and we first show that the problem is equivalent to quantum error correction using group-covariant encodings.  We prove a no-go theorem showing that that no finite dimensional, group-covariant quantum codes exist for Lie groups with an infinitesimal generator (\emph{e.g.,} U(1), SU(2), and SO(3)). We then explain how one can circumvent this no-go theorem using infinite dimensional codes, and we give an explicit example of a covariant quantum error correcting code using continuous variables for the group U(1). Finally, we demonstrate that all finite groups have finite dimensional codes, giving both an explicit construction and a randomized approximate construction with exponentially better parameters.
\end{abstract}
\maketitle
%%%%%%%%%%%%%%%%%%%%%%%%%%%%%%%%%%%%%%%%%%%%%%%%%%%%%%%%%%%%%%%%%%%%%%%%%%%%%%

%%%%%%%%%%%%%%%%%%%%%%%%%%%%%%%%%%%%%%%%%%%%%%%%%%%%%%%%%%%%%%%%%%%%%%%%%%%%%%
\ssection{Introduction}%
One of Shannon's original insights in the formulation of information theory was to focus on the transmission of sequences of symbols, such as strings of $0$'s and $1$'s, without regard to the semantic content of the message. This approach makes it possible to encode an enormous variety of messages, from phone numbers to photos, as long as the original information can be faithfully represented in terms of a sequence of symbols. The same situation exists in the quantum world: quantum information theorists are primarily concerned with information that can be stored in a system of qubits (or larger quantum systems), independent of the type of information.
%This form of information is sometimes called \emph{speakable} information, a term originally coined by John Bell~\cite{bell2004speakable}.
\par
Here we study a situation in which the information is \emph{physical} and cannot be represented simply as abstract qubits. Consider the following purely classical scenario~\cite{gisin1999spin}. Alice wishes to transmit some directional information to Bob, \emph{e.g.}, the axis of rotation of a gyroscope indicated by the vector $\vec{n}$, so that Bob can prepare a gyroscope rotating around the same axis as Alice's. If Alice and Bob share a reference frame, Alice can measure different components of $\vec{n}$ and describe the result  \emph{in words} to Bob, who then prepares his own gyroscope to match. However, if Alice and Bob do not share a reference frame, \emph{i.e.}, if Alice and Bob do not know the relative alignment of their coordinate systems, then this task is impossible. Without a shared reference frame, Alice has no way to communicate a set of symbols to Bob indicating the axis of rotation of her gyroscope. Another simple example is clock synchronization, wherein  two distant observers want to synchronize their clocks, but it is not possible to do so by sending purely symbolic messages~\footnote{If the transmission time for each message is predetermined, that provides a resource that could itself be used for clock synchronization. To avoid this loophole, symbolic messages should not be received at predetermined times.}.
\par
Of course, the simple examples described above do not mean that sending physical information is impossible. For example, in a classical world, Alice can prepare and send a physical copy of her gyroscope to Bob, thereby indicating her direction. In this way, Alice and Bob can even establish a shared reference frame. Similarly, in the clock synchronization problem Alice can send a copy of her clock to Bob~\cite{preskill2000quantum} to establish a common time standard (ignoring relativistic effects).  Quantum mechanically, Alice can send direction information by sending polarized spins, while timing information can be sent by quantum clocks such as two-level atoms. As is common in the quantum world, many interesting and counterintuitive effects occur. For example, sending two anti-parallel spins polarized along the desired direction is a better direction indicator than sending two parallel spins~\cite{massar1995optimal,gisin1999spin}. The problem of aligning quantum reference frames has garnered significant attention in recent years~\cite{bagan2001optimal,jozsa2000quantum,souza2008quantum,bartlett2006degradation,bartlett2003classical,peres2001entangled,bartlett2007reference,massar1995optimal,gisin1999spin,preskill2000quantum,gour2008resource,marvian2014modes}.  
\par
In this paper we are interested in quantum error correction of physical information. Crucially, physical information can only be communicated using systems that themselves have the physical property of interest.  This places constraints on the actions that can be performed on the physical systems, since we can't, for example, destroy or change physical information arbitrarily.  In particular, there may be constraints on the set of possible encoding/decoding schemes that one might have used to make the system more robust to errors, thereby limiting our ability to perform quantum error correction. In this paper, we will characterize the constraints placed on quantum error correction of physical information.
\par
In each of the examples described above, Bob's lack of knowledge about Alice's reference frame or time standard is mathematically modeled by the action of an unknown element of some group on Alice's state.  For directional reference frames, Alice and Bob are related by an unknown rotation (\emph{i.e.}, an element of $SO(3)$), whereas in the example of clock synchronization Alice and Bob's clocks are related by an unknown time translation (which can be thought of as an element of $U(1)$).  In the spirit of~\cite{marvian2013theory} (which generalizes reference frame information to general asymmetry information in the context of resource theory~\cite{marvian2014asymmetry,marvian2013theory,marvian2014extending,brandao2013resource,veitch2014resource,devetak2008resource}), we study error correction of physical information that transforms under an arbitrary group $G$. An important reduction following from the analysis of~\cite{marvian2013theory} is that the existence of encoding schemes for this type of information is equivalent to the existence of ordinary, yet $G$-covariant, encoding schemes, which can correct the same errors. 
\par
In this paper, we first study the case in which the group $G$ has at least one infinitesimal generator (\emph{e.g.}, the rotation group and time translation examples above).  In this first case, we find a result strikingly different from conventional, abstract quantum information: we prove a \emph{no-go} theorem showing that it is impossible to encode physical information in any number of finite dimensional systems in such a way that the encoding allows for perfect correction of any erasure error. 
We then show that both conditions of the no-go theorem are necessary by constructing codes that circumvent the theorem when either of the conditions is violated.  Specifically, we first demonstrate how one can encode physical information to protect against erasure errors when one uses continuous-variable modes (with \emph{infinite dimensional Hilbert spaces}). Since continuous-variable modes are used, we expect this result to be of practical interest.{}
We then construct a perfect encoding scheme for any \emph{finite group} $G$ into finite dimensional spaces, which is again robust to erasure errors.
Finally, we study a family of group covariant random codes and show that they can provide encoding schemes with better parameters than the perfect schemes for finite groups.
\par
It is worth noting that the covariant channel formulation of the problem is closely related to other results in the literature that have very different motivations, including the \emph{Eastin-Knill Theorem}~\cite{EastinKnill} and recent studies of \emph{invariant perfect tensors}~\cite{li2016invariant}. We present a more detailed comparison in the~\hyperref[sec:discussion]{discussion}.

\ssection{Reference frame error correction}
We begin with a more formal description of error correction in the familiar case of spatial reference frames, which corresponds to $G = SO(3)$; the generalization to other groups is immediate.
%Setup:  
Suppose Alice and Bob share a (possibly noisy) quantum channel.  Alice wants to communicate some directional quantum information (a single spin, say) to Bob, but Alice and Bob do not share a common reference frame.  Specifically, their reference frames are related by an unknown rotation $R \in SO(3)$. 
Alice and Bob will claim success if Bob receives the spin in the same direction that it was sent by Alice (\emph{i.e.}, the directional information is unchanged -- a condition they could check at a later stage).
If the task is successful, Bob can use the received spin for various tasks, such as establishing a shared reference frame.
\par
The simplest method of sending any quantum information is to send the quantum state itself.  This is also true of directional information, but since the quantum channel between Alice and Bob is noisy, we must account for the possibility of error. 
We fix our error model to be \emph{erasure} of a single spin (or mode), and our goal will be to design an error correcting code to protect the directional information from this noise.
\par
For simplicity of presentation, we will discuss an encoding scheme that encodes one spin into three (see~\cref{fig:alicebob}), but the reader is cautioned that the choice of one-into-three is just for clarity; our result holds for an arbitrary one-to-many encoding.  We split the process into $6$ steps:
\begin{figure}
\includegraphics[height=3.2cm]{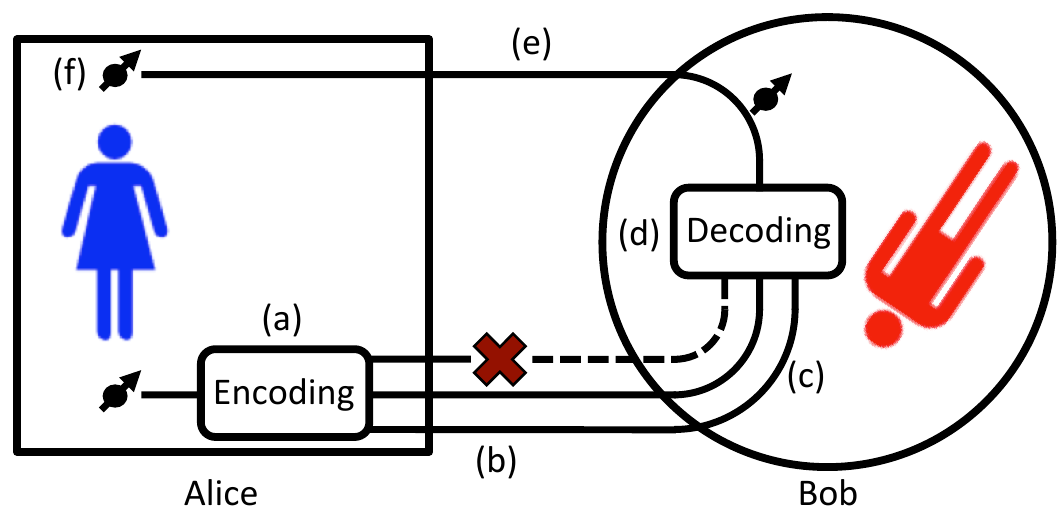}
\caption{ Setup: Alice wants to send a spin to Bob, but Alice and Bob do not share a reference frame. (a) Alice encodes her spin into an error correcting code. (b) The environment erases one of the spins. (c) Bob receives the encoded spins in \emph{his} reference frame. (d) Bob then decodes the remaining spins to reveal the original state. (e) Bob sends the decoded spin using a (hypothetical) perfect channel to Alice for verification. (f) Alice confirms that the recovered state is the same as her original state.}
\label{fig:alicebob}
\end{figure}
\begin{enumerate}
\item \emph{\Cref{fig:alicebob}a}.
Alice starts with an unknown input state $\rho_\text{in}$, a density operator on $\mathcal{H}_{\text{in}}$, representing some directional information. Alice encodes this initial state using an encoding channel $\Ep_A$. We use the subscript $A$ to indicate that $\Ep_A$ is the encoding map in Alice's reference frame, and to distinguish it from the map as seen in Bob's frame: $\Ep_B$, to which we will return shortly. Thus, the encoded state $\sigma_{123}$ on three spins is given by $\sigma_{123}= \Ep_A(\rho_\text{in})$.
\item \emph{\Cref{fig:alicebob}b}.
Spin $j \in \{1, 2, 3\}$ is lost, which is known as an erasure error. The erased spin could be any one of the three, but it is assumed that Bob can infer which.
\item \emph{\Cref{fig:alicebob}c}.
Prior to the erasure error, the encoded state as seen by Bob would be $U_1\ot U_2 \ot U_3 \sigma_{123} U_1^\dagger\ot U_2^\dagger \ot U_3^\dagger$, where $U_i=U_i(R)$ is a unitary representation of the unknown rotation $R$ mapping Alice's coordinate system Bob's. 
Bob then receives the state $\tr_j{( U_1\ot U_2 \ot U_3 \sigma_{123} U_1^\dagger\ot U_2^\dagger \ot U_3^\dagger)}$. 
\item \emph{\Cref{fig:alicebob}d}.
Bob decodes the state with an $R$-independent decoding map $\Dp_j$ to obtain ${\Dp_j \left[\tr_j(U_1\ot U_2 \ot U_3 \sigma_{123} U_1^\dagger\ot U_2^\dagger \ot U_3^\dagger)\right]}$. This is the state recovered in Bob's reference frame.  If the protocol is successful, this state should be equal to $\tilde \rho_\text{in}=U_\text{in} \rho_\text{in} U_\text{in}^\dagger$ in order to match Alice's original state, where $U_\text{in}=U_\text{in}(R)$ is the representation of the rotation group acting on the initial state, and the tilde signals that this is the input state as seen from Bob's rotated reference frame.

\item \emph{\Cref{fig:alicebob}e}. 
Bob sends the decoded state through a hypothetical perfect channel to Alice for verification.
\item \emph{\Cref{fig:alicebob}f}.
Success is claimed if the received state is the same as the initial state in Alice's frame.
\end{enumerate}
Using $\rho_\text{in}=U_\text{in}^\dagger \tilde \rho_\text{in}\, U_\text{in}$, the success condition becomes
\begin{equation}\label{eq:useless-success}
\tilde \rho_\text{in}\!=\! \Dp_j \big[\tr_j\!{\big( U_1\!\ot\! U_2 \!\ot\! U_3 \,\Ep_A(U_\text{in}^\dagger \tilde \rho_\text{in}U_\text{in}) U_1^\dagger\!\ot\! U_2^\dagger \!\ot\! U_3^\dagger\big)}\big],
\end{equation}
for all $R\in SO(3)$, states $\tilde \rho_\text{in}\in\mathcal{H}_{\text{in}}$ , and $j\in\{ 1,2,3\}$.
\par
\ssection{Covariant error correction}
Covariant quantum error correction is a seemingly different problem in which the encoding map is required to commute with the action of the group.  
Continuing the example of mapping a single spin into three, the covariance requirement is that the encoding map satisfy
\begin{equation}\label{eq:cov1}
U_1\ot U_2 \ot U_3\, \Ep(U_\text{in}^\dagger \rho_\text{in} \,U_\text{in}) U_1^\dagger\ot U_2^\dagger \ot U_3^\dagger=\Ep(\rho_\text{in})
\end{equation}
for all $R \in SO(3)$ and initial states $\rho_\text{in}$. To be clear, in this problem, Alice and Bob are assumed to share a single reference frame. Imposing the simple constraint (\ref{eq:cov1}) on the encoding map, however, defines an error correction problem equivalent to reference frame erasure correction.

Return now to the setting of reference frame error correction to see why. Alice performs the encoding $\Ep_A$ in her reference frame.  In Bob's reference frame, this operation is denoted by $\Ep_{B,R}$. ($\Ep_{B,R}$ is the quantum channel corresponding to the operation Alice performs as seen in Bob's reference frame.)  For a fixed $\Ep_A$ in Alice's reference frame, $\Ep_{B}$ in Bob's frame is still parametrized by the unknown rotation $R$, i.e., $\Ep_{B}=\Ep_{B,R}$.  Specifically, $\Ep_{B,R}(\tilde \rho_\text{in})={ U_1\ot U_2 \ot U_3 \,\Ep_A(\,U_\text{in}^\dagger\tilde \rho_\text{in} \,U_\text{in})\, U_1^\dagger\ot U_2^\dagger \ot U_3^\dagger}$.
The success condition simplifies to
\begin{equation}\label{eq:success} 
\Dp_j (\tr_j(\Ep_{B,R}( \tilde\rho_\text{in})))= \tilde\rho_\text{in},
\end{equation}
for all states $\tilde\rho_\text{in}$ and $j\in\{1,2,3\}$.
\par
Now introduce the \emph{average channel} $\Ep=\mathbb E_R[\Ep_{B,R}]$, where the average is over all rotations $R\in SO(3)$ according to the Haar measure.
By the linearity of the partial trace and the decoding channels, the error correction relation~(\ref{eq:success}) holds for the average channel: $\Dp_j (\tr_j(\Ep(\tilde\rho_\text{in})))= \tilde\rho_\text{in}$. Moreover, the averaged channel is clearly covariant in the sense of (\ref{eq:cov1}), provided we substitute $\tilde\rho_\text{in}$ for $\rho_\text{in}$ in the equation.
\par 
So if reference frame error correction~(\ref{eq:useless-success}) is possible, we have found a covariant erasure-correcting encoding. Moreover, it is straightforward to confirm that by choosing $\Ep$ to be $\Ep_A$, eqs.~(\ref{eq:success}) and~(\ref{eq:cov1}) lead to~eq.(\ref{eq:useless-success}). Therefore, reference frame error correction and covariant error correction are equivalent.

\par
\ssection{Results}
Let us now study a more general question. Consider an encoding map $\Ep$ which encodes an initial state on $\mathcal H_\text{in}$ into $n$ encoded systems on $\mathcal H_\text{out}=\mathcal H_1 \ot \cdots \ot \mathcal H_n$. We do not impose any constraints on the output Hilbert spaces at this point (i.e., they can be the same or different, finite or infinite dimensional, etc.) Suppose there exists a group $G$, and representations $U_\text{in}, U_1, \cdots,U_n$ acting on the different Hilbert spaces. Moreover, suppose that the channel is covariant under the action of the group:
\begin{equation}\label{eq:inv}
\Ep(\rho_\text{in})=U_1\ot \cdots \ot U_n \Ep(U_\text{in}^\dagger \rho_\text{in} U_\text{in}) U_1^\dagger \ot \cdots \ot U_n^\dagger
\end{equation}
Our goal is to answer the following question: \emph{is it possible to recover the original state after erasure of an arbitrary set of at most $k$ subsystems (which we henceforth refer to as \emph{modes)?}}
\par
We will study this question in different scenarios:
\begin{enumerate}
\item \emph{$G$ is a Lie group and the code is finite dimensional}. 
We prove a no-go theorem: no perfect covariant error correcting scheme can be implemented in this case. This applies to the example of sending spins, as in the original reference frame error correction task. In fact, the \emph{no-go} theorem applies to all groups with at least one infinitesimal generator, and it states that such generators can only act trivially on encoded states.
\item \emph{$G$ is a Lie group and the code is infinite dimensional}. 
We show that $G$-covariant error correcting codes are possible when the encoding uses infinite dimensional systems. This illustrates the existence of interesting error correcting codes for a Lie group when the conditions of the no-go theorem above are not satisfied. We provide an explicit code for $G=U(1)$ in~\cref{app:u1}.
\item \emph{$G$ is a finite group and the code is finite dimensional}. 
For any finite group $G$, we find examples of perfect covariant error correcting schemes. This is again consistent with our no-go theorem since finite groups do not have infinitesimal generators. We also provide a randomized construction in~\cref{app:apprand} to obtain approximate codes with better parameters.
\end{enumerate}
%%%%%%%%%%%%%%%%%%%%%%%%%%%%%%%%%%%%%%%%%%%%%%%%%%%%%%%%%%%%%%%%%%%%%%%%%%%%%%%%%%%%%%%%%%%%%%%%%%%%%%%%%%%%%%%%%%
\par
\ssection{Case 1: $G$ is a Lie group and the code is finite dimensional}\label{sec:nogo}
In this case, suppose that the local Hilbert space dimensions are finite, and that the group $G$ is a Lie group~\footnote{We exclude the case of $0$-dimensional Lie groups. Also, $G$ does not actually need to be a Lie group, but it must have at least one infinitesimal generator. In that case, our proof shows that the infinitesimal generators of the group can only act trivially on the system.}.
Choose one infinitesimal generator of the Lie group, without loss of generality. We denote this generator acting on the input mode by $T_\text{in}$ and on the $i$th output mode by $T_i$.
Thus, the generator acting on the full set of output modes is $T_\text{out}=  T_1+\cdots  T_n$.
Assume that $T_\text{in}$ is non-trivial; our goal will be to show that covariant quantum error correction is impossible with this assumption.
\par 
Consider an initial state $\rho_\text{in}$ and a slightly rotated state $\rho_\text{in}(\epsilon)= e^{-i\epsilon T_\text{in}}\rho_\text{in} e^{i\epsilon T_\text{in}}$. These states are encoded as $\sigma_\text{out}=\Ep(\rho_\text{in})$ and $\sigma_\text{out}(\epsilon) =\Ep(\rho_\text{in}(\epsilon))$.
Using the fact that $\Ep(\rho_\text{in})$ is invertible on its range, we can find a set of orthogonal isometries $\{E_i\}$, ($E_i^\dagger E_j=\delta_{ij}I$) and probabilities $p_i$ such that \[\Ep(\rho_\text{in})=\sum_i {p_i E_i \rho_\text{in} E_i^\dagger}.\] (see, \emph{e.g.},~\cite{nielsen2002quantum}, theorem 10.1 and its proof using $\mathcal H_\text{in}$ as the code space.)
The inverse channel $\Ep^{-1}(\sigma_\text{out})$ can be described by the same set of isometries on the range of $\Ep$ \[\Ep^{-1}( \rho_\text{out})=\sum_i {E_i^\dagger  \rho_\text{out} E_i}+ \Pi_\perp  \rho_\text{out} \Pi_\perp,\] where $\Pi_\perp=I-\sum_i{E_iE_i^\dagger}$. 
A crucial but elementary property of $\Ep^{-1}$ is that if $\sigma_\text{out} = \Ep (\rho_\text{in})$ and $A$ is some arbitrary operator, then $\Ep^{-1}(A\sigma_\text{out})=\Ep^\dagger(A)\rho_\text{in}$, where $\Ep^\dagger(A)=\sum_i{p_iE_i^\dagger A E_i}$.
Expanding the relation $\rho_\text{in}-\rho_\text{in}(\epsilon)=\Ep^{-1}( \sigma_\text{out} -  \sigma_\text{out}(\epsilon))$ to first order in $\epsilon$ we obtain 
\begin{align}\label{eq:3}
[T_\text{in},\rho_\text{in}]
&=\Ep^{-1}([T_\text{out}, \sigma_\text{out}]) \\
&=
[\Ep^\dagger(T_\text{out}),\rho_\text{in}].\nonumber
\end{align}
Under the assumption that error correction succeeds, we can then recover the original state from any of the $n-k$ subsets of the encoded modes. This means that upon tracing out all output modes except the $i$th mode, the remaining state $\rho_{i}$ is independent of the initial state (since if it weren't the mode number $i$ would contain information about the input state).
\par
Thus, for any state $\rho_\text{in}$, we find that $\tr(T_i \sigma_\text{out})=\alpha_i$, where $\alpha_i$ is independent of $\rho_\text{in}$. It is easy to see that \[\alpha_i=\tr(T_i \sigma_\text{out})=\tr( T_i \Ep(\rho_\text{in}))=\tr(\Ep^\dagger( T_{i})\rho_\text{in})\] for all $\rho_\text{in}$. Hence $\Ep^\dagger(T_i)\propto I$, and consequently $\Ep^\dagger(T_\text{out})\propto I$. This implies that the last term in~\cref{eq:3} is zero, which means that $[T_\text{in},\rho_\text{in}]=0$ for all $\rho_\text{in}$.  In order for $T_\text{in}$ to commute with all $\rho_\text{in}$ it must be trivial, which is a contradiction of our assumption.  We conclude that perfect recoverability is impossible.
%%%%%%%%%%%%%%%%%%%%%%%%%%%%%%%%%%%%%%%%%%%%%%%%%%%%%%%%%%%%%%%%%%%%%%%%%%%%%%%%%%%%%%%%%%%%%%%%%%%%%%%%%%%%%%%%%%
\par
\ssection{Case 2: $G$ is a Lie group and the code is infinite dimensional}
If we allow Alice the ability to use infinite dimensional Hilbert spaces (violating one of the hypotheses of our no-go theorem), then even a na\"ive solution to the problem exists.  Intuitively, a simple way to achieve the task is for Alice to append a classical gyroscope to the encoded state that she sends to Bob~\footnote{To be precise, each classical gyroscope determines one axis.  In order to send a classical reference frame we need at least two gyroscopes for the $x$ and $y$ axes. By ``gyroscope'' we mean a complete indicator of the reference frame.}.  Bob can then infer information about Alice's reference frame by measuring the state of the gyroscope, thereby establishing a common reference frame. Indeed, this is one strategy we will outline below. Since the full state is sent through the noisy channel, Alice actually sends two gyroscopes in order to safeguard against loss of one of the encoded shares. Any reader disappointed by the construction's use of effectively classical gyroscopes should be heartened to know that the 1-into-3 encoding described in \cref{app:u1} achieves covariant error correction without them.

In the reference frame error correction paradigm, Alice chooses her favourite (non-covariant) erasure code and appends two redundant ancilla (the classical gyroscopes) indicating her reference frame to the encoded state.  The ancilla must necessarily be states in infinite dimensional Hilbert spaces so that the no-go theorem does not apply (and in this protocol this is also necessary so that Alice can specify her reference frame with perfect precision)\footnote{Some readers might take issue with calling this an erasure code, since such codes are usually constructed such that the Hilbert spaces of each share are the same (so that all Hilbert spaces in this case must then be infinite dimensional).  If desired, one can simply embed finite dimensional Hilbert spaces into the infinite dimensional spaces such that the group acts on these subspaces according to the associated finite dimensional representation and trivially on the rest.}. If any shares of the erasure code are lost, Bob can first measure the gyroscopes to learn Alice's reference frame, and then use the standard decoding on the remaining shares in the right frame.  Since Alice sent two ancilla, one can freely be lost without failure.

Let us now study this problem in the covariant quantum error correction paradigm.  Let $\mathcal H_G=\text{span}\{\ket g\}$, where $g\in G$.  The group acts via $U(g)\ket h= \ket {gh}$ \footnote{In fact, it is not necessary to assume that the basis is indexed by group elements -- they can be indexed by any set on which the group acts faithfully.}. To encode her state, Alice chooses her favourite, non-covariant erasure correcting code (denoted by $\Ep_0$), such as the $\mathbb{C}^3\to(\mathbb{C}^3)^{\otimes 3}$ qutrit code for example (wlog).  As before, we define the rotated encoding map (i.e., the map in Bob's frame) by
\begin{equation}\label{eq:append}
 \Ep_{g} (\Psi)= U(g)^{\ot 3} \Ep_0 \left[U^{\dagger} (g)\Psi U(g)\right] U^{\dagger{\ot 3}}(g). 
\end{equation}
To complete the encoding, Alice appends two ancilla in the state $\ket{e}\bra{e}$ (where $e \in G$ is the identity element) for a full encoded state $\Ep_{0}(\Psi) \ot \ket e\bra e^{\ot 2}$ as seen in her frame.  The two $\ket{e}\bra{e}$ registers represent the classical gyroscopes above.  The encoding is made \emph{covariant} by averaging over the group $G$.  Thus, the full encoding is defined by symmetrizing the channel and ancilla together: \[\Ep(\Psi)=\int_{g\in G}{d g\,\,\, \Ep_{g}(\Psi) \ot \ket g\bra g^{\ot 2}}, \] which can be easily seen to be covariant.

Our decoding procedure is then fairly simple: one need only measure any ancilla that are not lost, collapsing the state to one corresponding to the measured group element. We can then recover the encoded state from the any qutrit shares they were not erased by the noisy channel.

The procedure described above is not the only method one can use in this case.  In \cref{app:u1} we describe an explicit, group covariant, continuous-variable quantum erasure code for the example of $G=U(1)$.  An input continuous variable mode is mapped into three physical modes via the encoding
\begin{equation*}
E_{U(1)}=\sum_{x,y\in \mathbb Z}{ \ket{-3y, -x+y, 2(y+x)}_{123}\bra {x}_\text{in}}.
\end{equation*}
We leave all relevant details for \cref{app:u1}.

\ssection{Case 3: $G$ is a finite group and the code is finite dimensional}\label{sec:finitegroup}
Consider a \emph{finite} group $G$. Here we show that it is possible to find $G$-covariant channels that encode the input Hilbert space into finite dimensional Hilbert spaces while satisfying the erasure correction conditions.% which encode the initial state into some finite dimensional codes.
\par
Suppose the group $G$ acts on some set $A$. By definition, the action of $G$ permutes the elements of $A$. Our goal is to construct an error correction scheme for which the action of the group commutes with the process of encoding, erasure, and decoding.
To achieve our goal, we first start with a non-covariant erasure .  We then consider a tensor product of many copies of this non-covariant code, one tensor factor for each element of $A$. As it  happens, this code (defined using many copies of a non-covariant code) is already a covariant code!  To see this, note that the encoding acts as a tensor product over the factors, while the group action simply permutes the factors. Therefore, the encoding map and the group action commute, which implies that the encoding is $G$-covariant.%\REMOVE{ the action of the group permutes the tensor factors. Becuase  With this assumption, one can see that the error correcting proccess commutes with the action of the group. }
\par
To be more precise, consider a channel $\Ep_0 : S(\mathcal H_\text{in}) \rightarrow S( \mathcal H_\text{out}:= \mathcal H^{\ot n})$ where $S(\mathcal H)$ denotes the space of density matrices on the Hilbert space $\mathcal H$. Suppose that $\Ep_0$ is an encoding map that allows for recovery after erasure of an arbitrary set of $k$ of the $n$ output modes. However, we make no assumptions about the covariance of $\Ep_0$ -- it is an arbitrary erasure correcting map. We now introduce a new encoding
\par
\begin{equation*}
\Ep=\bigotimes_{a\in A} \Ep_0= \Ep_0^{\otimes |A|},\qquad \Ep : S(\mathcal H_\text{in}^{\ot |A|}) \rightarrow S(\mathcal H_\text{out}^{\ot |A|}),
\end{equation*}
where we have used $\bigotimes_{a\in A} \Ep_0$ to indicate that the different tensor copies are labeled by elements of $A$.
For each $g \in G$ the action of the representation on $\mathcal H^{\ot |A|}$ is defined by \[U(g)\ket {\phi_{a_1}} \ket {\phi_{a_2}}\cdots \ket{\phi_{a_{|A|}}} = \ket{\phi_{g^{-1}a_1}}\ket{\phi_{g^{-1}a_2}}\cdots \ket{\phi_{g^{-1}a_{|A|}}}.\] Here $a_1 \cdots a_{|A|}$ is a list of the elements of $A$. The covariance of $\Ep$ follows from the definition, and the error correction properties of $\Ep$ are directly inherited from those of $\Ep_0$.  Therefore, we have succeeded in finding a perfect $G$-covariant channel.
\Cref{fig:permutation} shows an example in which $G=S_3$ (the permutation group on $3$ elements) and $A=\{1,2,3\}$.
\par
\begin{figure}
\includegraphics[height=3.2cm]{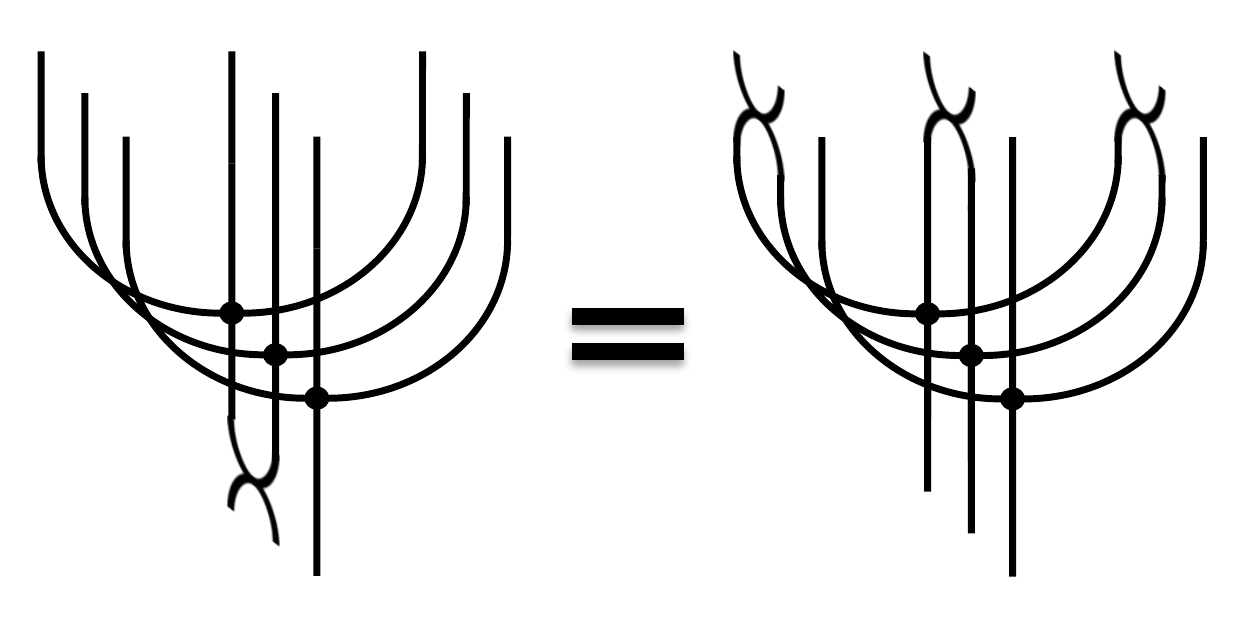}
\caption{Permutation covariance for the group $S_3$ acting on $S_3$ (i.e., $G=A=S_3$).  Each fork represents a code that maps one qudit into three, and can correct an erasure error on any one output qudit.  $\pi_{12}\in G$ is the transposition that swaps systems $1$ and $2$. \emph{Left.} The map $\Ep(U_\text{in}(\pi_{12})\rho_\text{in}U_\text{in}(\pi_{12})^\dagger)$. The group action permutes the inputs to the channel. \emph{Right.} The map $U_\text{out}(\pi_{12})\Ep(\rho_\text{in})U_\text{out}(\pi_{12})^\dagger$. As it is evident from the wiring of the forks, these two maps are equivalent.}
\label{fig:permutation}
\end{figure}
\par
While our construction can be formally extended to infinite groups with their associated infinite dimensional representations, we have not determined which additional conditions need to be imposed in order for the argument to remain mathematically rigorous.
\par
The construction presented in this section provides codes in which the Hilbert spaces can be exponentially large in $|G|$. However, it is known that in many cases random codes give near optimal error correcting schemes with good parameters~\cite{shor2002quantum,devetak2005private,lloyd1997capacity,hayden2008decoupling,hamada2005information}. In~\cref{app:apprand}, we show that choosing a random covariant isometry yields approximate error correcting codes for which the dimension of each mode is just $|G|$. 
For these codes, the worst-case fidelity of recovery, $F_\text{worst}$, behaves well with high probability. Specifically, $\text{P}(F_\text{worst}<1-\epsilon)$ decays exponentially in $|G|$. For example, we will show that: 
\begin{multline}\label{eq:conc}
\text{P} \left(F_\text{worst}<1-|G|^{\frac{9-2n}{8}} \right)\leq\\
 \exp\left(-\frac{|G|^2}{216} \left[|G|^{\frac{2n-8}{4}}-432\log\left(30|G|^{\frac{7+2n}{8}} \right) \right] \right)
\end{multline}
It is clear that for $n\geq 5$ and $|G|$ sufficiently large, the exponent on the right-hand side becomes arbitrarily negative, indicating that the worst-case fidelity of recovery is very close to $1$ with very high probability. 
                              
\ssection{Discussion}\label{sec:discussion}
We showed that perfect error correction of physical information against erasure is a process that depends on the details of the symmetry group and dimensions of the code.  For example, covariant error correction is impossible when the symmetry group is a Lie group and the code is finite dimensional. This is connected to the following \emph{no-go} theorems in the literature:
\par
$\bullet$ \emph{Eastin-Knill theorem}~\footnote{We thank Beni Yoshida pointing out the connection to the Eastin-Knill Theorem.}. 
Eastin and Knill proved~\cite{EastinKnill} that it is not possible to encode information in an error-detecting code in such a way that a set of universal gates can be implemented transversally. 
We can reproduce the main thrust of the Eastin-Knill theorem~\footnote{The Eastin-Knill theorem also discusses the possibility of encoding information in the disconnected components of the Lie group, a point that is absent in our work. Furthermore, the full Eastin-Knill theorem makes reference to universal gates.  In order to fully reproduce the theorem, we would need additional arguments concerning continuity of the channel and error detection.} using an instance of our no-go theorem in which the input space is the set of $N$ logical qudits, the output consists of physical qudits, and letting the group be $G=U(N)$.  Moreover, our continuous variable, infinite dimensional code construction provides a demonstration that the Eastin-Knill theorem can be circumvented in principle, although our examples do not appear to be useful for fault-tolerant quantum computation.
\par
$\bullet$ \emph{Invariant perfect tensors}. A quantum state on the tensor product of a number of Hilbert spaces is a \emph{perfect tensor} if, for any bipartition of the Hilbert space into two collections of constituent factors, it forms an isometry from the smaller space to the larger~\cite{pastawski2015holographic}.
Motivated by the construction of physical states in the Hilbert space of loop quantum gravity, the authors in~\cite{li2016invariant} defined the notion of \emph{invariant perfect tensors} as those perfect tensors which are invariant with respect to the action of $SU(2)$. In~\cite{li2016invariant}, 
it was proved that there are no invariant perfect tensors with four tensor factors. This can be seen as a direct consequence of our no-go theorem for $G=SU(2)$
, by considering a four-partite invariant perfect tensor as a $1$ mode to $3$ mode isometry.  Such an invariant perfect tensor with 4 tensor factors would define an $SU(2)$-covariant erasure correcting code, which is prohibited by our no-go theorem. Furthermore, our no-go theorem states that there are no invariant perfect tensors with higher numbers of tensor factors, thereby solving an open question in~\cite{li2016invariant}.
\par
One might hope to find a more quantitative relation between some measure of the size of the group and the dimension of the code when error correction is possible. For example, 
a condition of the form $|G|\leq \text{dim(code)}$ (i.e., dimension of the physical Hilbert space) is consistent with our no-go theorem and the examples in Cases 1 and 2. Another interesting avenue for future research relates to approximate error correction, in which one might like to find a relationship between the error tolerance $\epsilon$, group size $|G|$, and dimension of the code.

\par\ssection{Acknowledgements}
We thank Dawei Ding, Iman Marvian, Michael Walter, and Beni Yoshida for helpful discussion. SN acknowledges support from Stanford Graduate Fellowship. GS acknowledges support from a NSERC postgraduate scholarship.  This work was supported by the CIFAR and the Simons Foundation.

\bibliography{cecc}
\clearpage

\onecolumngrid
\appendix
\section{$G = U(1)$ and the code is continuous-variable}%Detailed error correction procedure for $U(1)$ code }
\label{app:u1}
Here we provide an explicit $U(1)$-covariant $1\rightarrow 3$ encoding.
The construction presented in this section does not violate the no-go theorem stated in \emph{Case 1} above as the local systems are infinite dimensional. Since the symmetry group in question is $U(1)$, this code could be implemented in optical modes, and it is arguably more natural than the construction presented in \emph{Case 2}.
\par
We take the Hilbert space to be the space of functions on a circle using the position basis $\{\ket \phi\}_{\phi\in [0,2\pi)}$.
$U(1)$ acts on this space via the regular representation: if $g=e^{i\theta}\in U(1)$, then the action of the regular representation is defined by $U(g)\ket \alpha= \ket {\alpha+\theta}$.
It is convenient to work in the Fourier basis where the Hilbert space is described by the conjugate momentum basis $\{\ket n\}_{n\in \mathbb Z}$ and the group acts by $U(g)\ket n=e^{i n \theta} \ket n$.
We define the isometry to be the following operator expressed in the conjugate momentum basis
\begin{align*}
E_{U(1)}=\sum_{x,y\in \mathbb Z}{ \ket{-3y, -x+y, 2(y+x)}_{123}\bra {x}_\text{in}}.
\end{align*}
More explicitly, the isometry maps the state $\sum_x{\phi(x)\ket x}_\text{in}$ to $\ket \Psi_{123}=\sum_{x,y}{\phi(x)\ket{-3y,-x+y,2(y+x)}}_{123}$. It is easy to see that this isometry is $U(1)$-covariant:
\begin{align*}
U(g)^{\ot 3} E_{U(1)} U(g)^\dagger
&=e^{i(-3y-x+y+2(y+x))}E_{U(1)}e^{-ix} \\
&=E_{U(1)}
\end{align*}

Here we show, step by step, that this mapping can correct an erasure error. 
Consider the encoded density matrix
\begin{align*}
\Psi_{123}=\sum_{x_1,y_1,x_2,y_2\in \mathbb Z}{\phi(x_1)\phi(x_2)^* \ket{-3y_1, -x_1+y_1, 2(y_1+x_1)}\bra{-3y_2, -x_2+y_2, 2(y_2+x_2)}}_{123}.
\end{align*}
We will study the loss of modes $1$, $2$ and $3$, in turn.
\begin{enumerate}
\item \emph{Loss of the first mode}.
The resulting density matrix is
\begin{align*}
\Psi_{23}=\sum_{x_1,x_2,y\in \mathbb Z}{\phi(x_1)\phi(x_2)^* \ket {-x_1+y, 2(y+x_1)}\bra{-x_2+y, 2(y+x_2)}_{23}}.
\end{align*}
Decoding starts with the linear map $\ket{a,b} \rightarrow \ket{a,b-2a}$, yielding
\begin{align*}
\sum_{x_1,x_2,y\in \mathbb Z}{\phi(x_1)\phi(x_2)^* \ket {x_1+y,4x_1}\bra{x_2+y, 4x_2}_{23}}.
\end{align*}
We then use an isometry which maps the states of the form $\ket {a,4b}$ to $\ket {a,b}$
\begin{align*}
\sum_{x_1,x_2,y\in \mathbb Z}{\phi(x_1)\phi(x_2)^* \ket {x_1+y,x_1}\bra{x_2+y, x_2}_{23}}.
\end{align*}
Finally by $\ket{a,b}\rightarrow \ket {a-b,b}$, we obtain
\begin{align*}
\sum_{x_1,x_2,y\in \mathbb Z}{\phi(x_1)\phi(x_2)^* \ket {y,x_1}\bra{y, x_2}_{23}}.
\end{align*}
Therefore, tracing out mode $2$ reveals the original state.
\item \emph{Loss of the second mode}.
The resulting density matrix is
\begin{align*}%y2=x1+y1-x2
\Psi_{13}=\sum_{x_1,y,x_2\in \mathbb Z}{\phi(x_1)\phi(x_2)^* \ket{-3y,  2(y+x_1)}\bra{-3(-x_1+y+x_2), 2(-x_1+y+2x_2)}}_{13},
\end{align*}
or, equivalently by the change of variable $y\rightarrow y+x_1$,
\begin{align*}%y2=x1+y1-x2
\Psi_{13}=\sum_{x_1,y,x_2\in \mathbb Z}{\phi(x_1)\phi(x_2)^* \ket{-3(y+x_1),  2(y+2x_1)}\bra{-3(y+x_2), 2(y+2x_2)}}_{13}.
\end{align*}
We now use an isometry which maps states of the form $\ket {3a,2b}$ to $\ket {a,b}$ 
\begin{align*}%y2=x1+y1-x2
\sum_{x_1,y,x_2\in \mathbb Z}{\phi(x_1)\phi(x_2)^* \ket{-(y+x_1),  (y+2x_1)}\bra{-(y+x_2), (y+2x_2)}}_{13}.
\end{align*}
By $\ket{a,b}\rightarrow \ket {a,2a+b}$, we have
\begin{align*}%y2=x1+y1-x2
\sum_{x_1,y,x_2\in \mathbb Z}{\phi(x_1)\phi(x_2)^* \ket{-(y+x_1),  -y}\bra{-(y+x_2), y}}_{13}.
\end{align*}
We now use $\ket{a,b}\rightarrow \ket {-(a+b),b}$ to obtain
\begin{align*}%y2=x1+y1-x2
\sum_{x_1,y,x_2\in \mathbb Z}{\phi(x_1)\phi(x_2)^* \ket{x_1,  -y}\bra{x_2, y}}_{13}.
\end{align*}
Tracing out mode $3$ reveals the original state.
\item \emph{Loss of the third mode}.
Again, the resulting density matrix is
\begin{align*}%y2=y1-x1+x2
\Psi_{12}=\sum_{x_1,y,x_2\in \mathbb Z}{\phi(x_1)\phi(x_2)^* \ket{-3y, -x_1+y}\bra{-3(y+x_1-x_2), -2x_2+y+x_1}}_{12}.
\end{align*}
Using the change of variable $y \rightarrow y+x_1$ we have
\begin{align*}%y2=y1-x1+x2
\Psi_{12}=\sum_{x_1,y,x_2\in \mathbb Z}{\phi(x_1)\phi(x_2)^* \ket{-3(y-x_1), -2x_1+y}\bra{-3(y-x_2), -2x_2+y}}_{12}.
\end{align*}
Applying an isometry that maps $\ket {3a,b}$ to $\ket {a,b}$ yields
\begin{align*}%y2=y1-x1+x2
\sum_{x_1,y,x_2\in \mathbb Z}{\phi(x_1)\phi(x_2)^* \ket{-(y-x_1),- 2x_1+y}\bra{-(y-x_2), -2x_2+y}}_{12}.
\end{align*}
Using $\ket{a,b} \rightarrow \ket{a,a+b}$,
\begin{align*}%y2=y1-x1+x2
\sum_{x_1,y,x_2\in \mathbb Z}{\phi(x_1)\phi(x_2)^* \ket{-(y-x_1), -x_1}\bra{-(y-x_2), -x_2}}_{12}.
\end{align*}
Finally the isometry $\ket{a,b} \rightarrow \ket{a+b,-a}$ turns the state to
\begin{align*}%y2=y1-x1+x2
\sum_{x_1,y,x_2\in \mathbb Z}{\phi(x_1)\phi(x_2)^* \ket{-y, x_1}\bra{-y, x_2}}_{12}.
\end{align*}
Thus we can recover the state on mode $2$.
\end{enumerate}                                    

\section{$G$ is a finite group and the code is a random $G$-covariant isometry}\label{app:apprand}
In the construction presented for \emph{Case 3}, the local Hilbert space dimension can grow exponentially with $|G|$.  In this section we present an alternative, approximate method for error correction in which the local Hilbert space dimemsions are equal to $|G|$.
Our goal will be to prove~\cref{eq:conc}.
\par
Consider a $1\rightarrow n$ encoding. We will look for isometries that map $\mathcal H_{G} \rightarrow  \mathcal H_{G}^{\ot n}$, where $\mathcal H_G$ denotes the Hilbert space associated to the regular representation of $G$ with the basis $\{\ket g\}_{g\in G}$. Thus $\dim \mathcal H_G= |G|=d$. We represent the action of the regular representation of $g\in G$ on $\mathcal H_G$ by $U(g)$. 
\par
To construct a random covariant map, we start with a random invariant state $\ket \Psi \in \mathcal H_G^{\ot (n+1)}$. For our purposes: a random state is one that is chosen randomly with respect to the unitary invariant measure; random unitaries are unitaries chosen randomly with respect to the Haar measure; and a state is invariant if $U(g)^{\ot (n+1)} \ket \Psi = \ket \Psi$ for all $g\in G$. By projecting our chosen state onto an un-normalized, maximally entangled state $\ketum_{AB}=\sum{\ket i_A\ket i_B}$ we obtain a map $E$ (which is close to an isometry w.h.p.) from $\mathcal H_\text{in}\rightarrow \mathcal H^{\ot n}$,
\[ E_{\text{in},1\cdots n}=\sqrt d \braum_{\text{in},0} \ket \Psi_{0\cdots n}.\]
Note that the covariance of $E$ defined by $U(g)^{\ot n}  E =  E U(g)$, which follows from the invariance of $\ket \Psi$. From $E$ we can define the exact isometry $T$ as
\[T:=E ( E^\dagger E)^{-1/2}.\]
One can verify that $T$ is also a covariant map, since $[E^\dagger E,U(g)]$ for all $g\in G$. Our encoding is then defined by \[\Ep(\rho_\text{in})= T \rho_\text{in}T^\dagger.\]

With the covariant encoding in hand, we now turn our attention to the decoding.  Before diving in, let us first define two notational conventions that will be used frequently henceforth. Firstly, we will use $\tr_{\hat x}$ to indicate tracing out all subsystems \emph{except} the set $x$. Secondly, if there are two isomorphic Hilbert spaces $\mathcal H_\alpha$ and $\mathcal H_\beta$ with the same preferred basis, and if the operator $X_\alpha$ acts on $\mathcal H_\alpha$, then by $(X_\alpha)_{\beta}$ we mean the operator $X_\alpha$ acting on $\mathcal H_\beta$ (in the sense that the matrix corresponding to $X_\alpha$ is simply applied to $\mathcal H_\beta$). One can think of $(X_\alpha)_\beta$ as overriding the Hilbert space indices. When it is clear to do so, we use $X_\beta$ instead of $(X_\alpha)_\beta$ for brevity.
\par
To decode after loss of one of the modes, say mode $1$ without loss of generality, Bob first replaces the lost mode by a maximally mixed state $\tau_1$ and then decodes the state $\tau_1 \ot \tr_1\left[ \Ep (\rho_\text{in})\right]$.  The decoding map is given by
\begin{align*}
\sigma_\text{out}=\Dp_1(\rho_{12\cdots n})=\left(\tr_{\hat 2}\left[(U_{23}^T V_{23\cdots n})\rho_{12\cdots n}(U_{23}^T V_{23\cdots n})^\dagger \right]\right)_\text{out},
\end{align*}
where $U_{01}$, and $V_{23\cdots n}$ are unitaries that transform $\ket \Psi_{0\cdots n}$ into its Schmidt form: \[U_{01}\ot V_{2\cdots n} \ket \Psi_{0\cdots n}= \sum_{i,j}\sqrt{\lambda_{ij}} \ket {ij}_{01} \ot \ket{ij0\cdots 0}_{23\cdots n}, \]
and $U_{23}=(U_{01})_{23}$ is the same operator as $U_{01}$ but acting on the Hilbert spaces indexed by $2$ and $3$.  In other words, $U_{01}=\left( U_{23}\right)_{01}$.
\par 
With the decoding above, our task is now to prove~\cref{eq:conc}.  Our first step will be bounding the worst-case fidelity of recovery $F_\text{worst}$ in terms of the distance between $\Psi_{01}$ (the reduced density matrix of the invariant state $\ket \Psi$) and the maximally mixed state:
\begin{lem}\label{lem:lemma1}
For and $0\leq \epsilon \leq 1$, if $\left\|\Psi_{01}-\tau_{01}\right\|_\infty\leq \frac\epsilon {3d^2}$, then $1-\epsilon\leq F_\text{worst}$.
\end{lem}
\begin{proof}
We will prove in three steps:
\begin{itemize}
\item \textbf{Step 1.} We first simplify the expression for the recovered state and show that \[ \Dp_1(\tau_1\ot \Ep(\rho_\text{in}))=\tr_1\left({{\Psi_{01}^{T}}^{1/2}}{\Psi_0^T}^{-1/2}(\rho_\text{in})_0 {\Psi_0^T}^{-1/2} {\Psi_{01}^{T}}^{1/2}\right) .\]
\item \textbf{Step 2.} We then use joint concavity of the fidelity, and properties of the Schatten norm to bound the worst-case fidelity \[F_\text{worst}\geq \min_{\ket\kappa}{\left|\bra{\kappa}_0 \tr_1 {\left( \Psi_{01}^{1/2}\right)}\left( \frac{\Psi_0^{-1/2}}{\sqrt d}\right) \ket{\kappa}_0\right|} .\] 
From the above equation, it is already clear that if $\Psi_0$ and $\Psi_{01}$ are close to the maximally mixed state, then the worst-case fidelity will be close to $1$. We quantify this in the last step.
\item \textbf{Step 3.} We show that for $0\leq \epsilon \leq 1$, if $\left\|\Psi_{01}-\tau_{01}\right\|_\infty\leq \frac\epsilon {3d^2}$, then $1-\epsilon\leq F_\text{worst}$.
\end{itemize}
\par\ssection{\textbf{Step 1}}\par
We begin with the expression for the recovered state,
\begin{align}\label{eq:dec_enc_app}
\Dp_1(\tau_1\ot \Ep(\rho_\text{in}))=\tr_{\hat 2}\left(U_{23}^T V_{23 \cdots n} \tr_1(T\rho_\text{in} T^\dagger) V_{23 \cdots n}^\dagger U_{23}^*\right).
\end{align} 
Using the fact that $E^\dagger E=d(\Psi_0^T)_\text{in}$, and the definition $\tilde \rho_\text{in}= \frac{1}{d} \left({\Psi_0^T}^{-1/2}\right)_\text{in}\rho_\text{in}\left( {\Psi_0^T}^{-1/2}\right)_\text{in}$, we have that $T \rho_\text{in} T^\dagger=E\tilde \rho_\text{in} E^\dagger$ . 
From the definition of $E$ we can simplify the formula for the encoding map:
\begin{align*}
\Ep(\rho_\text{in}) = E \tilde \rho_\text{in} E^\dagger =d \tr_0 \left({\ket \Psi\bra \Psi_{0 \cdots n} \tilde \rho_0^T}  \right)
\end{align*}
Therefore,
\begin{align} \label{eq:dec_semi_ex}
\Dp_1(\tau_1\ot \Ep(\rho_\text{in}))=d\tr_{\hat 2}\left[U_{23}^T V_{2\cdots n} \ket \Psi \bra \Psi_{0\cdots n} V^\dagger_{2\cdots n} U^{*}_{23} \tilde \rho_0^T\right].
\end{align}
However, recall that $U_{01}\ot V_{2\cdots n} \ket \Psi_{0\cdots n}= \sum_{i,j} \sqrt{\lambda_{ij}}\ket {ij}_{01} \ot \ket{ij0\cdots 0}_{23\cdots n}$, and that $U_{01} \Psi_{01}^{1/2}U_{01}^\dagger= \sum_{ij}\sqrt{\lambda_{ij}}\ket{ij}\bra{ij}$.  Thus we obtain 
\[V_{2\cdots n} \ket \Psi_{0\cdots n} = \Psi_{01}^{1/2}U_{01}^\dagger \ket \phi^+_{02} \ket \phi^+_{13}\ket{0\cdots 0}_{4 \cdots n}= U_{23}^* \left({\Psi_{01}^T}^{1/2}\right)_{23} \ket {\phi^+}_{02} \ket {\phi^+_{13}}\ket{0\cdots 0}_{4 \cdots n}. \]
Using~\cref{eq:dec_semi_ex}, we find 
\begin{align*}
&\Dp_1(\tau_1\ot \Ep(\rho_\text{in}))\\
&=d\tr_{\hat 2}\left( \left({\Psi_{01}^{1/2}}^T\right)_{23}\ket{\phi^+}_{02}\ket{\phi^+}_{13}\bra{\phi^+}_{02}\bra{\phi^+}_{13} \left({\Psi_{01}^{1/2}}^T\right)_{23} \tilde \rho_0^T\right)\\
&=d\tr_{3}\left( \left({\Psi_{01}^{1/2}}^T\right)_{23} \left(\tilde \rho_0\right)_2  \left({\Psi_{01}^{1/2}}^T\right)_{23} \right)\\
&=\tr_1\left({{\Psi_{01}^{T}}^{1/2}}{\Psi_0^T}^{-1/2}\rho_0 {\Psi_0^T}^{-1/2} {\Psi_{01}^{T}}^{1/2}\right).
\end{align*}
Therefore, we have achieved the goal of step 1.
\par
\ssection{\textbf{Step 2}}\par
Our goal now is to lower bound the fidelity of recovery. Since the fidelity is jointly concave, we know that the minimum fidelity of recovery for the channel is achieved with a pure input state, say $\rho_0=(\ket \kappa \bra \kappa)^T$, where we have added the transpose to simplify the expressions. In this case, the recovered state takes the following form: 
\begin{equation*}
\Dp_1(\tau_1\ot \Ep(\rho_\text{in}))=\tr_1\left({\Psi_{01}^{1/2}}{\Psi_0}^{-1/2}\ket \kappa_0\bra \kappa_0  {\Psi_0}^{-1/2} {\Psi_{01}^{1/2}}\right)^T,
\end{equation*}
so that the minimum fidelity is
\begin{align*}
F_{min}=\min_{\ket\kappa}{\sqrt{\tr\left({\bra \kappa_0 \Psi_{01}^{1/2} \Psi_0^{-1/2}\ket \kappa_0 \bra \kappa_0 \Psi_0^{-1/2}\Psi_{01}^{1/2}\ket \kappa_0}\right)}}=\min_{\ket \kappa}\left({ \left\|\bra \kappa_0 \Psi_{01}^{1/2} \Psi_0^{-1/2}\ket \kappa_0\right\|_2}\right).
\end{align*}
To proceed, we use the following basic property of the Schatten norm: for $\frac1p+\frac1q=1$, $\| Y\|_p \geq  |\tr (X Y^\dagger)|$ if $\|X\|_q = 1$. Applying this inequality when $X=I_1/\sqrt d$ and $p=q=2$ we find:
\begin{align}\label{eq:step2}
\left\|\bra \kappa \Psi_{01}^{1/2} \Psi_0^{-1/2}\ket \kappa\right\|_2
&=\max \left\{\left| \tr\left(X_1\bra {\kappa}_0 \Psi_{01}^{1/2} \Psi_0^{-1/2}\ket {\kappa}_{0}\right)\right| \big|\,\,\,\| X \|_2=1\right\}\nonumber \\
&\geq \frac{1}{\sqrt d}\left| \tr\left(\bra {\kappa}_0 \Psi_{01}^{1/2} \Psi_0^{-1/2}\ket {\kappa}_{0}\right)\right|\nonumber \\
&= \left|\bra{\kappa}_0 \tr_1 {\left( \Psi_{01}^{1/2}\right)}\left( \frac{\Psi_0^{-1/2}}{\sqrt d}\right) \ket{\kappa}_0\right|.
\end{align}
This conludes step 2.
\par\ssection{\textbf{Step 3}}\par
We would ultimately like to lower bound the worst-case fidelity using concentration of measure techniques for $\Psi_{01}$ and $\Psi_0$. 
\par
We start by upper bounding $\left \| \tr_1 \left(\Psi_{01}^{1/2}\right) - I_0 \right \|_\infty$ and $\left \|\left(\frac{\Psi_{0}^{-1/2}}{\sqrt d}\right) -  I_{0} \right \|_\infty$, assuming that $\left \| \Psi_0 - \tau_0 \right \|_\infty \leq \frac{1}{2d}$.
\begin{enumerate}
\item Upper bound for $\left \| \tr_1 \left(\Psi_{01}^{1/2}\right) - I_0 \right \|_\infty$:
\par
\begin{align*}
\left \| \tr_1 \left(\Psi_{01}^{1/2}\right) - I_0 \right \|_\infty
&=\left \| \tr_1 \left(\Psi_{01}^{1/2}-\frac{I_{01}}{d}\right) \right \|_\infty
=\max_{\ket \alpha} \left| \bra \alpha_0\tr_1 \left(\Psi_{01}^{1/2}-\frac{I_{01}}{d}\right)\ket \alpha_0\right|\\
&\leq \max_{\ket \alpha}\sum_{g \in G} \left| \bra \alpha_0\bra g_1 \tr_1 \left(\Psi_{01}^{1/2}-\frac{I_{01}}{d}\right)\ket \alpha_0\ket g_1\right|\\
&\leq d\left \|\Psi_{01}^{1/2} - \frac{I_{01}}{ d} \right \|_\infty,
\end{align*}
where $\ket g$, $g \in G$ form a basis for evaluating the trace, the first inequality is the triangle inequality, and the second inequality comes from the fact that the infinite Schatten norm of a Hermitian operator is equal to its maximum eigenvalue. Now, one can check that for any $\lambda \geq 0$, $|\lambda^{1/2} - 1/d| \leq d|\lambda - 1/d^2|$. Taking $\{\lambda_i\}$ to be the set of eigenvalues of $\Psi_{01}^{1/2}$, and using the aforementioned inequality, we obtain
\begin{align}\label{eq:bound1}
\left \|\Psi_{01}^{1/2} - \frac{I_{01}}{d} \right \|_\infty= \max_{i} {\left| \lambda_i^{1/2} - \frac1d\right|} \leq d\max_{i} {\left| \lambda_i - \frac{1}{d^2}\right|}\leq d \left \| \Psi_{01}- \tau_{01} \right \| _\infty.
\end{align}
Thus \[ \left \| \tr_1 \left(\Psi_{01}^{1/2}\right) - I_0 \right \|_\infty \leq d^2 \left \| \Psi_{01}- \tau_{01} \right \| _\infty.\]
\item Upper bound for $\left \|\left(\frac{\Psi_{0}^{-1/2}}{\sqrt d}\right) -  I_{0} \right \|_\infty$:
\par 
One can simply check that for any real number $\lambda$ such that $|\lambda - 1/d| \leq 1/{2d}$, then $\left| \lambda^{-1/2}/{\sqrt d} -1 \right|\leq d\left|\lambda - {1}/{d} \right|$. In particular, since this inequality holds for all of the eigenvalues of $\Psi_0$, we can derive the following bound for the operator norm:
\begin{align}\label{eq:bound2}
\left \|\left(\frac{\Psi_{0}^{-1/2}}{\sqrt d}\right) -  I_{0} \right \|_\infty \leq d \left \| \Psi_0 - \tau_0 \right \|_\infty
\end{align} 
\end{enumerate}
To proceed, we \emph{assume} that $\left \| \Psi_{01} - \tau_{01} \right \|_\infty \leq \frac{ \epsilon}{3d^2}$ for $0\leq \epsilon \leq 1$.  Combining this assumption with~\cref{eq:step2},
 we have
\begin{align*}
& \left|\bra{\kappa}_0 \tr_1 {\left( \Psi_{01}^{1/2}\right)}\left( \frac{\Psi_0^{-1/2}}{\sqrt d}\right) \ket{\kappa}_0\right|\nonumber\\
&= \left|1 + \bra {\kappa}_0\left[ \tr_1 \left(\Psi_{01}^{1/2}\right) - I_0 \right] \ket{\kappa}_0  + \bra {\kappa}_0\left[\left(\frac{\Psi_{0}^{-1/2}}{\sqrt d}\right) -  I_{0}\right] \ket{\kappa}_0  + \bra {\kappa}_0\left[ \tr_1 \left(\Psi_{01}^{1/2}\right) - I_0 \right]\left[\left(\frac{\Psi_{0}^{-1/2}}{\sqrt d}\right) -  I_{0}\right] \ket{\kappa}_0 \right|\nonumber\\
&\geq 1 -\left| \bra {\kappa}_0\left[ \tr_1 \left(\Psi_{01}^{1/2}\right) - I_0 \right] \ket{\kappa}_0\right|  - \left| \bra {\kappa}_0\left[\left(\frac{\Psi_{0}^{-1/2}}{\sqrt d}\right) -  I_{0}\right] \ket{\kappa}_0 \right| - \left| \bra {\kappa}_0\left[ \tr_1 \left(\Psi_{01}^{1/2}\right) - I_0 \right]\left[\left(\frac{\Psi_{0}^{-1/2}}{\sqrt d}\right) -  I_{0}\right] \ket{\kappa}_0 \right|\nonumber,
\end{align*}
where the inequality in the last line is the triangle inequality.  Now, since $\bra \kappa X \ket \kappa \leq \|X\|_\infty$ for any matrix $X$, we have
\begin{align*}
 \left|\bra{\kappa}_0 \tr_1 {\left( \Psi_{01}^{1/2}\right)}\left( \frac{\Psi_0^{-1/2}}{\sqrt d}\right) \ket{\kappa}_0\right|\nonumber
&\geq 1 -\left\|  \tr_1 \left(\Psi_{01}^{1/2}\right) - I_0 \right\|_\infty  - \left\| \left(\frac{\Psi_{0}^{-1/2}}{\sqrt d}\right) -  I_{0}\right\|_\infty - \left\|\left[ \tr_1 \left(\Psi_{01}^{1/2}\right) - I_0 \right]\left[\left(\frac{\Psi_{0}^{-1/2}}{\sqrt d}\right) -  I_{0}\right] \right\|_\infty\nonumber\\
&\geq 1- \left| \tr_1 \left(\Psi_{01}^{1/2}\right) - I_0 \right \|_\infty -\left \|\left(\frac{\Psi_{0}^{-1/2}}{\sqrt d}\right) -  I_{0} \right \|_\infty-\left \| \tr_1 \left(\Psi_{01}^{1/2}\right) - I_0 \right \|_\infty \left \|\left(\frac{\Psi_{0}^{-1/2}}{\sqrt d}\right) -  I_{0} \right \|_\infty \nonumber,
\end{align*}
where the second inequality follows from the fact that $\|XY\|_\infty \leq \|X\|_\infty\|Y\|_\infty$ for any pair of matrices $X$ and $Y$. Using~\cref{eq:bound1,eq:bound2} above,
\begin{equation*}
\left|\bra{\kappa}_0 \tr_1 {\left( \Psi_{01}^{1/2}\right)}\left( \frac{\Psi_0^{-1/2}}{\sqrt d}\right) \ket{\kappa}_0\right|\nonumber
\geq 1- d^2 \left \| \Psi_{01}- \tau_{01} \right \| _\infty -d \left \| \Psi_0 - \tau_0 \right \|_\infty -\left(d^2\left \| \Psi_{01}- \tau_{01} \right \| _\infty\right) \left(d \left \| \Psi_0 - \tau_0 \right \|_\infty \right)\nonumber.
\end{equation*}
Note that the condition $\left \| \Psi_0 - \tau_0 \right \|_\infty \leq \frac{1}{2d}$ is satisfied, since $\left \| \Psi_0 - \tau_0 \right \|_\infty\leq d\left \| \Psi_{01}- \tau_{01} \right \| _\infty$ and $\left \| \Psi_{01} - \tau_{01} \right \|_\infty \leq \frac{ \epsilon}{3d^2}$. Finally, since $\left \| \Psi_0 - \tau_0 \right \|_\infty\leq d\left \| \Psi_{01}- \tau_{01} \right \| _\infty$, we have that
\begin{equation*}
\left|\bra{\kappa}_0 \tr_1 {\left( \Psi_{01}^{1/2}\right)}\left( \frac{\Psi_0^{-1/2}}{\sqrt d}\right) \ket{\kappa}_0\right|\nonumber
\geq 1- 2 d^2 \left \| \Psi_{01}- \tau_{01} \right \| _\infty- \left(d^2\left \| \Psi_{01}- \tau_{01} \right \| _\infty\right)^2,
\end{equation*}
and we therefore conclude that
\begin{equation*}
F_\text{worst}
\geq \left|\bra{\kappa}_0 \tr_1 {\left( \Psi_{01}^{1/2}\right)}\left( \frac{\Psi_0^{-1/2}}{\sqrt d}\right) \ket{\kappa}_0\right| 
\geq 1- 2 d^2 \left \| \Psi_{01}- \tau_{01} \right \| _\infty- \left(d^2\left \| \Psi_{01}- \tau_{01} \right \| _\infty\right)^2
\geq 1-\epsilon\,,
\end{equation*}
which proves the lemma.
\end{proof}                               
\par
To complete the proof, it remains to be shown that our assumption is valid.  Specifically, in order to show that the worst-case fidelity is close to $1$, it suffices to prove that the reduced density matrix of random invariant states, $\Psi_{01}$, is very close to the maximally mixed state in operator norm (\emph{i.e.}, $\|\Psi_{01}-\tau_{01}\|_\infty$ is small) with high probability.  Since
\begin{equation*}
\| \Psi_{01}-\tau_{01}\|_\infty = \max_{\sigma_{01}} {\left| \tr \left[\sigma_{01} (\Psi_{01} - \tau_{01})\right]\right|},
\end{equation*}
where the maximization is done over all possible density matrices $\sigma$, we can instead study the quantity on the right hand side. To show that this is small, we will follow the techniques used in~\cite{hayden2006aspects,harrow2004superdense,hayden2004randomizing,bennett2005remote}. 
\par
Before stating the proof in its full glory, let us first gain an imprecise, high-level overview of the strategy.  We will first define an $\epsilon$-net on the set of density matrices on $\mathcal H_0 \ot \mathcal H_1$, \emph{i.e.,} a finite set of density matrices $\tilde\sigma_{01}$ such that any other density matrix $\sigma_{01}$ is close to one of the elements of the net in the trace norm. If we can then show that $\left| \tr \left[\tilde \sigma_{01} (\Psi_{01} - \tau_{01})\right]\right|$ is small for every $\tilde \sigma$ in the net, then it must be small for \emph{all} density matrices $\sigma_{01}$. Using large deviation methods, we will then prove that for any fixed density matrix $\sigma_{01}$ (including the elements of the net), $\left| \tr \left[\sigma_{01} (\Psi_{01} - \tau_{01})\right]\right|$ is small with very high probability. Since the number of elements in the net is finite (with a known upper bound), we can then use a union bound to show that $\left| \tr \left[\sigma_{01} (\Psi_{01} - \tau_{01})\right]\right|$ is small for all elements in the net with high probability.  Therefore, we can bound $\left| \tr \left[\sigma_{01} (\Psi_{01} - \tau_{01})\right]\right|$, arriving at our desired conclusion. 
\par
We will now we give a detailed proof of~\cref{eq:conc}. Let $\mathcal P_{\delta,\sigma_{01}}$ be the probability that, for a fixed $\sigma_{01}$, $\left| \tr \left[\sigma_{01} (\Psi_{01} - \tau_{01})\right]\right|\geq \delta/d^2$, and let $\mathcal P_\delta=\max_{\sigma_{01}} \mathcal P_{\delta,\sigma_{01}}$, where the maximum is over all density matrices on $\mathcal H_0 \ot \mathcal H_1$.  
The following lemma relates $\text{P}\left(\left\| \Psi_{01}- \tau_{01} \right\|_\infty \leq \frac{\epsilon}{3d^2}\right)$ to $\mathcal P_{\delta}$.
\begin{lem}\label{lem:lemma2}
For $0\leq \alpha \leq \epsilon$, we have 
\begin{equation*}
\text{P}\left(\left\| \Psi_{01}- \tau_{01} \right\|_\infty \leq \frac{\epsilon}{3d^2}\right)\geq 1- \mathcal P_{\frac{\epsilon-\alpha}3}\cdot\left[ \frac{15d^2}{\alpha}\right]^{2d^2}.
\end{equation*}
\end{lem}
\begin{proof}
Consider an $\frac{\alpha}{3d^2}$-trace distance net $\mathcal M$ of pure states in $\mathcal H_0\ot \mathcal H_1$, with $\alpha \leq\epsilon$.
For every pure state $\sigma_{01}$, there exists a pure state $\tilde \sigma _{01}$ such that
\begin{equation}\label{eq:netdef}
\| \sigma_{01}- \tilde \sigma_{01}\|_1 \leq \frac{\alpha}{3d^2},
\end{equation} 
by definition.  It is known that we can choose $\mathcal M$ such that $|\mathcal M|\leq \left(\frac{15d^2}{\alpha} \right)^{2d^2}$~\cite[Lemma II.4]{hayden2004randomizing}. 
Now if $\left|\tr(\tilde \sigma_{01}[\Psi_{01}-\tau_{01}]) \right|\leq \frac{\epsilon - \alpha}{3d^2}$, then from~\cref{eq:netdef} it follows that
\begin{align*}
\Big|\tr(\sigma_{01}[\Psi_{01}-\tau_{01}]) \Big|
&\leq \Big|\tr(\tilde \sigma_{01}[\Psi_{01}-\tau_{01}]) \Big| + \Big|\tr((\sigma_{01}-\tilde \sigma_{01})[\Psi_{01}-\tau_{01}]) \Big|\\
&\leq \frac{\epsilon - \alpha}{3d^2} + \| \sigma_{01} - \tilde \sigma_{01} \|_1 \|\Psi_{01}- \tau_{01}\|_\infty \\
&\leq \frac{\epsilon - \alpha}{3d^2} +\| \sigma_{01} - \tilde \sigma_{01} \|_1\\
&\leq \frac{\epsilon}{3d^2}.
\end{align*}
Therefore, 
\begin{align}\label{eq:medsteps}
\text{P}\left(\left\| \Psi_{01}- \tau_{01} \right\|_\infty \leq \frac{\epsilon}{3d^2} \right)\nonumber
&=\text{P}\left(\forall \sigma_{01} : \Big|\tr\left( \sigma_{01}[\Psi_{01}- \tau_{01}]\right) \Big| \leq \frac{\epsilon}{3d^2} \right)\nonumber\\
&\geq\text{P}\left(\forall \tilde \sigma_{01}\in \mathcal M : \Big|\tr\left( \tilde\sigma_{01}[\Psi_{01}-\tau_{01}]\right) \Big| \leq \frac{\epsilon-\alpha}{3d^2} \right)\nonumber\\
&=1- \text{P}\left(\exists \tilde \sigma_{01}\in \mathcal M : \Big|\tr\left( \tilde\sigma_{01}[\Psi_{01}-\tau_{01}]\right)\Big| \geq \frac{\epsilon-\alpha}{3d^2} \right).
\end{align}
We can simplify~\cref{eq:medsteps} using a union bound:
\begin{align*}
\text{P}\left(\exists \tilde \sigma_{01}\in \mathcal M : \Big|\tr\left( \tilde\sigma_{01}[\Psi_{01}-\tau_{01}]\right)\Big| \geq \frac{\epsilon-\alpha}{3d^2} \right)
\leq \sum_{\tilde \sigma_{01} \in \mathcal M} \mathcal P_{\frac{\epsilon-\alpha}3,\tilde \sigma_{01}}
\leq \mathcal P_{\frac{\epsilon-\alpha}3} \cdot  |\mathcal M| 
\end{align*}
This, along with~\cref{eq:medsteps}, conclude the proof of the lemma.
\end{proof}
In~\cref{app:tr_conc}, we will use large deviation techniques to show that% for $0\leq \delta \leq 1$,  
\begin{equation}\label{eq:appres}
\mathcal P_{\delta} \leq \exp \left(-d^{n-2}\delta^2/6 \right)\qquad\text{for}\qquad 0\leq \delta \leq 1. 
\end{equation}
We will defer the proof to \cref{app:tr_conc} but use the result immediately.  Combining~\cref{lem:lemma1},~\cref{lem:lemma2} and~\cref{eq:appres} we have
\begin{align*}
\text{P}\left(F_\text{worst} \leq 1-\epsilon \right)
&\leq \text{P}\left(\left\| \Psi_{01}- \tau_{01} \right\|_\infty \geq \frac{\epsilon}{3d^2} \right)
\leq \min_{0\leq \alpha \leq \epsilon}  \mathcal P_{\frac{\epsilon-\alpha}3} \cdot  \left[ \frac{15d^2}{\alpha}\right]^{2d^2} \\
&\leq \min_{0\leq \alpha \leq \epsilon}\exp\left(-d^{n-2} (\epsilon-\alpha)^2/54 +2d^2 \log \left( \frac{15d^2}{\alpha}\right)\right).
\end{align*}
One convenient choice of $\epsilon$ and $\alpha$ is $\epsilon =d^{\frac{9-2n}{8}}$ and $\alpha=\epsilon/2$. With this choice we find
\begin{align*}
\text{P}\left(F_\text{worst} \geq 1-\epsilon \right) \leq \exp\left(-\frac{d^2}{216} \left[d^{\frac{2n-8}{4}}-432\log\left(30d^{\frac{7+2n}{8}} \right) \right] \right),
\end{align*}
which reduces to~\cref{eq:conc} after substituting $|G|$ for $d$.

\section{Proof of~\cref{eq:appres}}\label{app:tr_conc}
The goal of this appendix is to prove~\cref{eq:appres}. The discussion is split into two parts: we first explain the random invariant state construction, and then we prove the desired bound.
\par
\ssection{Construction of random invariant states}
Consider the invariant subspace of $\mathcal H^{\ot (n+1)}$ -- it is easy to see that the invariant subspace is spanned by states of the form 
\begin{align*}
\frac{1}{\sqrt{d}}\sum_{g\in G}{\ket{gh_1, gh_2, \cdots, gh_n, g}}_{0\cdots n},
\end{align*}
We now introduce an isometry $M$ from $\mathcal H^{\ot n}$ to the invariant subspace of $\mathcal H^{\ot(n+1)}$,
\begin{align*}
M=\frac{1}{\sqrt{d}}\sum_{g,h_1,\cdots, h_n}{\ket{gh_1,gh_2,\cdots, gh_n,g}_{0\cdots n}\bra{h_1,h_2,\cdots, h_n}_{0 \cdots n-1}}.
\end{align*}
The projector onto the invariant subspace is defined as $\Pi_{0\cdots n} =M M^\dagger$. $\Pi_{0\cdots n}$ has the important property that, upon tracing out any one of the subsystems, it becomes the identity operator on the remaining subsystems.  That is
\begin{equation}\label{eq:traceout}
\tr_i \Pi_{0\cdots n} = I_{0\cdots \hat i \cdots n}.
\end{equation}
A random invariant state $\ket \Psi_{0\cdots n}$ is constructed by choosing a random state $\ket \phi_{0\cdots n-1}$ in $\mathcal H^{\ot n}$ from the unitary invariant measure, and then mapping $\ket \phi$ to $\mathcal H^{\ot (n+1)}$ using the isometry $M$, $\ket \Psi_{0\cdots n}= M \ket \phi_{0 \cdots n-1}$. 
\par
\ssection{Proof of~\cref{eq:appres}}
To begin, we will upper bound the moment generating function, $\mathbb E_{\Psi}\exp\left(t \tr{\left[\sigma_{01}\Psi_{01}\right]} \right)$, for an arbitrary density matrix $\sigma_{01}$, where $\Psi_{01}=\tr_{\hat0\hat1} \left[\Psi_{0\cdots n}\right]$ and the average is over random invariant states $\ket \Psi_{0\cdots n}$.
Note that $\tr{\left[\sigma_{01}\Psi_{01}\right]}=\tr{\left[\sigma_{01}\Psi_{0\cdots n}\right]}=\tr{\left[\sigma_{01}M\phi_{0\cdots n-1}M^\dagger\right]}=\tr{\left[M^\dagger\sigma_{01}M\phi_{0\cdots n-1}\right]}$. One can easily check that $M^\dagger \sigma_{01} M= \sigma^G_{01} \ot I_{2 \cdots n-1}$, where \[\sigma^G_{01}=\frac{1}{d} \sum_{g,h_1,h_2,h_1',h_2'}\ket{h_1,h_2}\bra{gh_1,gh_2}\sigma_{01} \ket{gh_1',gh_2'}\bra{h_1',h_2'}.\] 
One can also check that $\sigma^G_{01}$ is a density matrix, specifically a version of $\sigma_{01}$ symmetrized by the group $G$. Therefore, \[ \tr{\left[\sigma_{01}\Psi_{0\cdots n}\right]}= \bra{\phi} \sigma^G_{01} \ket \phi, \] where $\ket \phi=\ket \phi_{0\cdots n-1}$ is a state on $\mathcal H^{\ot n}$ chosen from the unitary invariant measure (see the first subsection of this appendix).
\par 
We now choose a Gaussian state $\ket g_{0 \cdots n-1}$ in which the coefficients of the wave function are chosen i.i.d from a complex Gaussian distribution centered at zero with variance $d^{-n}$. Thus $\mathbb E_{\ket g} \| g \|_2^2=1$. Therefore, we have
\begin{align}\label{eq:exprel}
\mathbb E_{\ket g} \exp\big(t\bra g \sigma^G_{01} \ket g \big)\nonumber
& =\mathbb E_{\ket \phi} \mathbb E_{\|g\|_2} \exp \big(t\,\| g\|_2^2\bra \phi \sigma^G_{01} \ket \phi \big)\nonumber \\
& \geq \mathbb E_{\ket \phi} \exp \Big(t\,\big[\mathbb E_{\|g\|_2}  \| g\|_2^2\big]\bra \phi \sigma^G_{01} \ket \phi \Big)\nonumber\\ 
& = \mathbb E_{\ket \phi} \exp \left(t\bra \phi \sigma^G_{01} \ket \phi \right)\nonumber\\
%\mathbb E_{\ket \phi} \exp\left(\right)   
&=\mathbb E_{\Psi}\exp\left(t \tr{\left[\sigma_{01}\Psi_{0\cdots n}\right]}\right),
\end{align}
where the inequality follows from the convexity of the exponential function. 
\par
Now suppose that the eigenvalues of $\sigma_{01}^G$ are $p_{i_0,i_1}$. Since the Gaussian states are unitarily invariant, we can evaluate $\mathbb E_{\ket g} \exp\left(t\bra g \sigma^G_{01} \ket g \right)$ in a basis in which $\sigma_{01}^G \ot I_{2\cdots n-1}$ is diagonal. 
In that basis, 
\begin{equation*}
\mathbb E_{\ket g} \exp\left(t\bra g \sigma^G_{01} \ket g \right)=\mathbb E_{\ket g} \exp\left(t\sum_{i_0 \cdots i_{n-1}}{ p_{i_0,i_1} |g_{i_0 \cdots i_{n-1}}|^2 }\right)= 
\prod_{i_0 \cdots i_{n-1}}{\mathbb E_{g_{i_0\cdots i_{n-1}}}\exp \left(t\, p_{i_0,i_1} |g_{i_0 \cdots i_{n-1}}|^2\right)}
\end{equation*}
However, the radial probability density for each coefficient is $p\,(|g_{i_0 \cdots i_{n-1}}|)=2d^n|g_{i_0 \cdots i_{n-1}}| \exp \left(-d^n |g_{i_0 \cdots i_{n-1}}|^2 \right)$. %Using this, we can exactly evaluate the expression. Using the formula 
Using the probability density formula, we find
\begin{align*}
\mathbb E_{g_{i_0\cdots i_{n-1}}}\exp \left(t p_{i_0,i_1} |g_{i_0 \cdots i_{n-1}}|^2\right)=\frac{1}{1-\frac{t\, p_{i_0,i_1}}{d^n}} \qquad \text{for}\qquad t\leq d^n/p_{i_0,i_1}.
\end{align*}
Assuming $t\leq d^n$, we have
\begin{align*}
\mathbb E_{\ket g} \exp\left(t\bra g \sigma^G_{01} \ket g \right)=\prod_{i_0,i_1} \Big(1-\frac{t\,p_{i_0,i_1}}{d^n}\Big)^{-d^{n-2}}.
\end{align*}
Ultimately, we will fix the value of $t$ to prove the bound in~\cref{eq:appres}, but we need to  distinguish the cases in which $t$ is positive or negative to bound the fluctuations of $\tr{[\sigma_{01}\Psi_{01}]}$ above or below $1/d^2$. Therefore, we discuss these two different ranges for $t$ separately:
\begin{enumerate}
\item Positive $t$:

We will use the assumption that $t$ is positive to limit the fluctuations of $\tr{[\sigma_{01}\Psi_{01}]}$ \emph{above} $1/d^2$. Let $0<s<1$ be a fixed number, and restrict $t$ to $0\leq t \leq s d^n$. Under these conditions, we have,%it is easy to check that 
\begin{align*}
\left(1- \frac{t\, p_{i_0,i_1}}{d^n} \right)^{-1}\leq \left(1+\frac{1}{1-s} \frac{t\,p_{i_0,i_1}}{d^n} \right).
\end{align*}
Therefore, 
\begin{align*}
\left(1- \frac{t\, p_{i_0,i_1}}{d^n} \right)^{-d^{n-2}}\leq \left(1+\frac{1}{1-s} \frac{t\,p_{i_0,i_1}}{d^n} \right)^{d^{n-2}} \leq \exp\left(\frac{1}{1-s}t\,p_{i_0,i_1}d^{-2} \right).
\end{align*}
Combining with~\cref{eq:exprel}, we have
\begin{align}\label{eq:post}
\mathbb E_{\ket g} \exp\big(t\bra g \sigma^G_{01} \ket g \big)\leq \mathbb E_{\ket g} \exp\left(t\bra g \sigma^G_{01} \ket g \right)=\prod_{i_0,i_1} \left(1-\frac{p_{i_0,i_1}t}{d^n}\right)^{-d^{n-2}}\leq \exp\left(\frac{1}{1-s}td^{-2} \right).
\end{align}
To bound the probabilities, we use Bernstein's trick:
\begin{align*}
\text{P}\left(\tr{\left[\sigma_{01}\Psi_{01}\right]}\geq \frac{1}{d^2}+ \frac{\delta}{d^2} \right)
&=\text{P}\left(
\exp\left(t\tr{\left[\sigma_{01}\Psi_{01}\right]}\right)\geq \exp\left(t\frac{1+\delta}{d^2}\right)
 \right) \\
&\leq \Big[\mathbb E_{\Psi} \exp\big(t\tr{\left[\sigma_{01}\Psi_{01}\right]}\big)\Big] \exp\left(-t\frac{1+\delta}{d^2}\right)\\ 
&\leq\exp\left( -td^{-2}\left(1+\delta - \frac{1}{1-s}\right)\right),
\end{align*}
where we used Markov's inequality for the exponentials and~\cref{eq:post}. To obtain the best result, we now set $t=sd^n$ and $s=1-(1+\delta)^{-1/2}$. With this substitution, 
\begin{align*}
\text{P}\left(\tr{\left[\sigma_{01}\Psi_{01}\right]}\geq \frac{1}{d^2}+ \frac{\delta}{d^2} \right)\leq \exp\left(-d^{n-2}(\sqrt{1+\delta}-1)^2 \right)\leq \exp\left(-d^{n-2}\delta^2/6 \right)
\end{align*}
where the last inequality is valid for $0\leq \delta \leq 1$.
\item Negative $t$:

We now use the constraint on $t$ to limit the fluctuations of $\tr{[\sigma_{01}\Psi_{01}]}$ \emph{below} $1/d^2$. Assuming that $s>0$ and $-sd^n \leq t\leq 0$, one can show that 
\begin{align*}
\left( 1-\frac{tp_{i_0,i_1}}{d^n}\right)^{-d^{n-2}}\leq \exp\left(t\frac{\log (1+s)}{s}p_{i_0,i_1}d^{-2} \right).
\end{align*}
Therefore, 
\begin{equation*}
\mathbb E_{\ket g} \exp\left(t\bra g \sigma^G_{01} \ket g \right)
=\prod_{i_0,i_1} \left(1-\frac{p_{i_0,i_1}t}{d^n}\right)^{-d^{n-2}}
\leq \prod_{i_0,i_1} \exp\left(t\frac{\log (1+s)}{s}p_{i_0,i_1}d^{-2} \right)
\leq \exp\left(\frac{\log (1+s)}{s}td^{-2} \right).
\end{equation*}
Thus,
\begin{align*}
\text{P}\left(\tr{\left[\sigma_{01}\Psi_{01}\right]}\leq \frac{1}{d^2}- \frac{\delta}{d^2} \right)
&= \text{P}\left(t\tr{\left[\sigma_{01}\Psi_{01}\right]}\geq t\left(\frac{1}{d^2}- \frac{\delta}{d^2}\right) \right) \\
&=\text{P}\left(\exp\big(t\tr{\left[\sigma_{01}\Psi_{01}\right]}\big)\geq \exp \left(t\frac{1-\delta}{d^2}\right) \right)\\
&\leq \Big[\mathbb E_{\Psi} \exp\left(t\tr{\left[\sigma_{01}\Psi_{01}\right]}\right)\Big]\exp \left(-td^{-2}(1-\delta) \right)\\
&\leq\exp \left(-td^{-2}\left(1-\delta-\frac{\log(1+s)}{s}\right) \right).
\end{align*}
We now fix $t=-sd^n$ and $s=\delta/(1-\delta)$ to get
\begin{align*}
\text{P}\left(\tr{\left[\sigma_{01}\Psi_{01}\right]}\leq \frac{1}{d^2}- \frac{\delta}{d^2} \right)\leq \exp\left[ d^{n-2}\left(\delta +\log(1-\delta)\right)\right]\leq \exp\left(-d^{n-2} \delta^2/2 \right)\leq \exp\left(-d^{n-2} \delta^2/6 \right).
\end{align*}
This concludes the proof of~\cref{eq:appres}.
\end{enumerate}
\end{document}